\documentclass[reqno]{amsart}
\usepackage{amssymb,graphics,graphicx,bbm,hyperref,color}

\usepackage{multicol}
\usepackage{enumitem}

\definecolor{gray}{rgb}{0.93,0.93,0.93}
\definecolor{light-gold}{rgb}{0.99,0.97,0.78}

\def\be{\begin{equation}}
\def\ee{\end{equation}}
\def\bm{\begin{multline}}
\def\bfig{\begin{figure}[htb]}
\def\efig{\end{figure}}
\newcommand{\dd}{{\rm d}}
\newcommand{\e}[1]{\,{\rm e}^{#1}\,}
\newcommand{\ii}{{\rm i}}
\def\Tr{{\operatorname{Tr\,}}}
\def\tr{{\operatorname{tr\,}}}

\numberwithin{equation}{section}
\newtheorem{theorem}{Theorem}%[section]

\newcommand{\caE}{{\mathcal E}}
\newcommand{\caH}{{\mathcal H}}
\newcommand{\caL}{{\mathcal L}}

\newcommand{\bbC}{{\mathbb C}}
\newcommand{\bbE}{{\mathbb E}}
\newcommand{\bbN}{{\mathbb N}}
\newcommand{\bbP}{{\mathbb P}}

\newcommand{\bsh}{{\boldsymbol h}}

\newcommand{\Id}{\mathrm{\texttt{Id}}}

\makeatletter
  \def\tagform@#1{\maketag@@@{\tiny{(#1)}\@@italiccorr}}
\makeatother

\renewcommand{\eqref}[1]{(\ref{#1})}

\begin{document}

%{\hfill\small \version} \vspace{2mm}

\title{Quantum spin correlations and random loops}

\author{Daniel Ueltschi}
\address{Department of Mathematics, University of Warwick,
Coventry, CV4 7AL, United Kingdom}
\email{daniel@ueltschi.org}


%\keywords{Random loop model, quantum Heisenberg models}

\begin{abstract}
We review the random loop representations of T\'oth and Aizenman-Nachtergaele for quantum Heisenberg models. They can be combined and extended so as to include the quantum $XY$ model and certain SU(2)-invariant spin 1 systems. We explain the calculations of correlation functions.

\end{abstract}

\thanks{Work partially supported by EPSRC grant EP/G056390/1.}
\thanks{\copyright{} 2014 by the author. This paper may be reproduced, in its
entirety, for non-commercial purposes.}

\maketitle

\section{Introduction}
\label{sec intro}

Random loop approaches to quantum spin systems offer an elegant and different perspective to quantum correlations. They find their origin in Feynman-Kac representations of quantum lattice systems. Motivated by earlier work of Conlon and Solovej \cite{CS}, T\'oth introduced a representation of the $S=\frac12$ ferromagnetic Heisenberg model that is based on the random interchange model \cite{Toth}. It allowed him to propose a bound for the free energy (it has been improved recently by Correggi, Giuliani, and Seiringer \cite{CGS}, who have reached the best possible constant). A similar representation was introduced by Aizenman and Nachtergaele for the $S=\frac12$ antiferromagnet model \cite{AN} and certain models with higher spins. It allowed them to relate the one-dimensional quantum chain to two-dimensional Potts and random cluster models, yielding new insights on the quantum spin chain. This work was reviewed and extended in \cite{Nac1,Nac2}. See also \cite{GUW} for a pedagogical introduction. Recently, Bachmann and Nachtergaele used the representation in their study of the classification of gapped ground states \cite{BN}.

A synthesis of these two representations was proposed in \cite{Uel}. In the case $S=\frac12$, it applies to models that interpolate between the Heisenberg ferromagnetic and antiferromagnetic models such as the quantum $XY$ model. It also applies to certain SU(2)-invariant models of spin 1. Thanks to this representation, the existence of spin nematic long-range order was established in the model with $S=1$ in dimension $d\geq3$ \cite{Uel}. It also plays a r\^ole in the recent proof of emptyness formation of quantum spin chains of Crawford, Ng, and Starr \cite{CNS}.

This article reviews some of the material treated in \cite{Uel}, and also complements it. We consider the case of an external magnetic field, possibly disordered. We detail formul\ae{} for the matrix elements of the operator $\e{-\beta H}$ and use them to calculate correlation functions. Since there are few loop correlation functions, and seemingly more quantum spin correlation functions, the relations provide useful identities; these identities do not seem otherwise immediate.

\section{Quantum spin models}

Let $(\Lambda,\caE)$ be a graph, with $\Lambda$ the (finite) set of vertices and $\caE$ the set of edges. Given $S \in \frac12 \bbN$, the Hilbert space is
\be
\caH_{\Lambda} = \bigotimes_{x\in\Lambda} \caH_{x},
\ee
where each $\caH_{x}$ is a copy of $\bbC^{2S+1}$. The spin operators are denoted $S_{x}^{i}$, with $x\in\Lambda$ and $i=1,2,3$. They satisfy the commutation relations $[S_{x}^{1},S_{y}^{2}] = \ii \delta_{x,y} S_{x}^{3}$, and further relations obtained by cyclic permutations of the indices $1,2,3$. Recall that ``classical configurations'' $\sigma \in \{-S, \dots, S-1, S\}^{\Lambda}$ form a basis of $\caH_{\Lambda}$ where the operators $S_{x}^{3}$ are diagonal: Using Dirac's notation,
\be
S_{x}^{3} |\sigma\rangle = \sigma_{x} |\sigma\rangle.
\ee

We consider the three operators $T_{xy}$, $P_{xy}$, $Q_{xy}$ on $\caH_{\{x,y\}}$ (and their extensions on $\caH_{\Lambda}$ by identifying $T_{xy}$ with $T_{xy} \otimes \Id_{\Lambda\setminus\{x,y\}}$, etc...): 
\begin{itemize}
\item $T_{xy}$ is the transposition operator:
\be
T_{xy} |a,b\rangle = |b,a\rangle.
\ee
\item $P_{xy}$ is the operator
\be
P_{xy} = \sum_{a,b=-S}^{S} (-1)^{a-b} |a, -a \rangle \langle b, -b|.
\ee
Equivalently, the matrix coefficients of $P_{xy}$ are given by
\be
\langle a,b| P_{xy} |c,d\rangle = (-1)^{a-c} \delta_{a,-b} \delta_{c,-d}.
\ee
Notice that $\frac1{2S+1} P_{xy}$ is the projector onto the spin singlet.
\item $Q_{xy}$ is identical to $P_{xy}$ except for the signs:
\be
\langle a,b| Q_{xy} |c,d\rangle = \delta_{a,b} \delta_{c,d}.
\ee
\end{itemize}
These operators can be written in terms of spin operators. The form depends on the spin.
In the case $S=\frac12$, we have
\be
\begin{split}
&T_{xy} = 2 \bigl( S_{x}^{1} S_{y}^{1} + S_{x}^{2} S_{y}^{2} + S_{x}^{3} S_{y}^{3} \bigr) + \tfrac12, \\
&Q_{xy} = 2 \bigl( S_{x}^{1} S_{y}^{1} - S_{x}^{2} S_{y}^{2} + S_{x}^{3} S_{y}^{3} \bigr) + \tfrac12.
\end{split}
\ee
In the case $S=1$, we have
\be
\begin{split}
&T_{xy} = \vec S_{x} \cdot \vec S_{y} + (\vec S_{x} \cdot \vec S_{y})^{2} - 1, \\
&P_{xy} = (\vec S_{x} \cdot \vec S_{y})^{2} - 1.
\end{split}
\ee
Here, we used the notation $\vec S_{x} \cdot \vec S_{y} = S_{x}^{1} S_{y}^{1} + S_{x}^{2} S_{y}^{2} + S_{x}^{3} S_{y}^{3}$.

Let $\bsh = (h_{x})_{x\in\Lambda}$ denote external magnetic fields. We consider two distinct families of Hamiltonians, indexed by the parameter $u \in [0,1]$:
\begin{align}
&H_{\Lambda,\bsh}^{(u)} = -\sum_{\{x,y\} \in \caE} \Bigl( u T_{xy} + (1-u) Q_{xy} - 1 \Bigr) - \sum_{x\in\Lambda} h_{x} S_{x}^{3}, \label{def fam1} \\
&\tilde H_{\Lambda,\bsh}^{(u)} = -\sum_{\{x,y\} \in \caE} \Bigl( u T_{xy} + (1-u) P_{xy} - 1 \Bigr) - \sum_{x\in\Lambda} h_{x} S_{x}^{3}. \label{def fam2}
\end{align}

Let $Z^{(u)}(\beta,\Lambda,\bsh)$ and $\tilde Z^{(u)}(\beta,\Lambda,\bsh)$ denote the corresponding partition functions:
\begin{align}
&Z^{(u)}(\beta,\Lambda,\bsh) = \Tr_{\caH_{\Lambda}} \e{-\beta H_{\Lambda,\bsh}^{(u)}}, \\
&\tilde Z^{(u)}(\beta,\Lambda,\bsh) = \Tr_{\caH_{\Lambda}} \e{-\beta \tilde H_{\Lambda,\bsh}^{(u)}}.
\end{align}

The Hamiltonians of Eqs \eqref{def fam1} and \eqref{def fam2} can also be expressed in terms of spin operators. In the case $S=\frac12$, we have
\be
H_{\Lambda,\bsh}^{(u)} = -2 \sum_{\{x,y\}\in\caE} \bigl( S_{x}^{1} S_{y}^{1} + (2u-1) S_{x}^{2} S_{y}^{2} + S_{x}^{3} S_{y}^{3} - \tfrac14 \bigr) - \sum_{x\in\Lambda} h_{x} S_{x}^{3}.
\ee
The case $u=1$ is the Heisenberg ferromagnet. The case $u=\frac12$ is the quantum XY model. If the graph is bipartite, the case $u=0$ is unitarily equivalent to the Heisenberg antiferromagnet.

In the case $S=1$, we have
\be
\label{def Ham spin 1}
\tilde H_{\Lambda,\bsh}^{(u)} = -\sum_{\{x,y\}\in\caE} \Bigl( u \vec S_{x} \cdot \vec S_{y} + (\vec S_{x} \cdot \vec S_{y})^{2} - 2 \Bigr) - \sum_{x\in\Lambda} h_{x} S_{x}^{3}.
\ee
It is well-known that any two-body SU(2)-invariant interaction for $S=1$ can be be written as $J_{1} \vec S_{x} \cdot \vec S_{y} + J_{2} (\vec S_{x} \cdot \vec S_{y})^{2}$. The phase diagram of this model is very interesting and it has been investigated by several authors \cite{BO,TZX,TLMP,FKK}. It is displayed in Fig.\ \ref{fig phd}. The line $J_{2}=0$ corresponds to the usual Heisenberg models.

\bfig
\includegraphics[width=9cm]{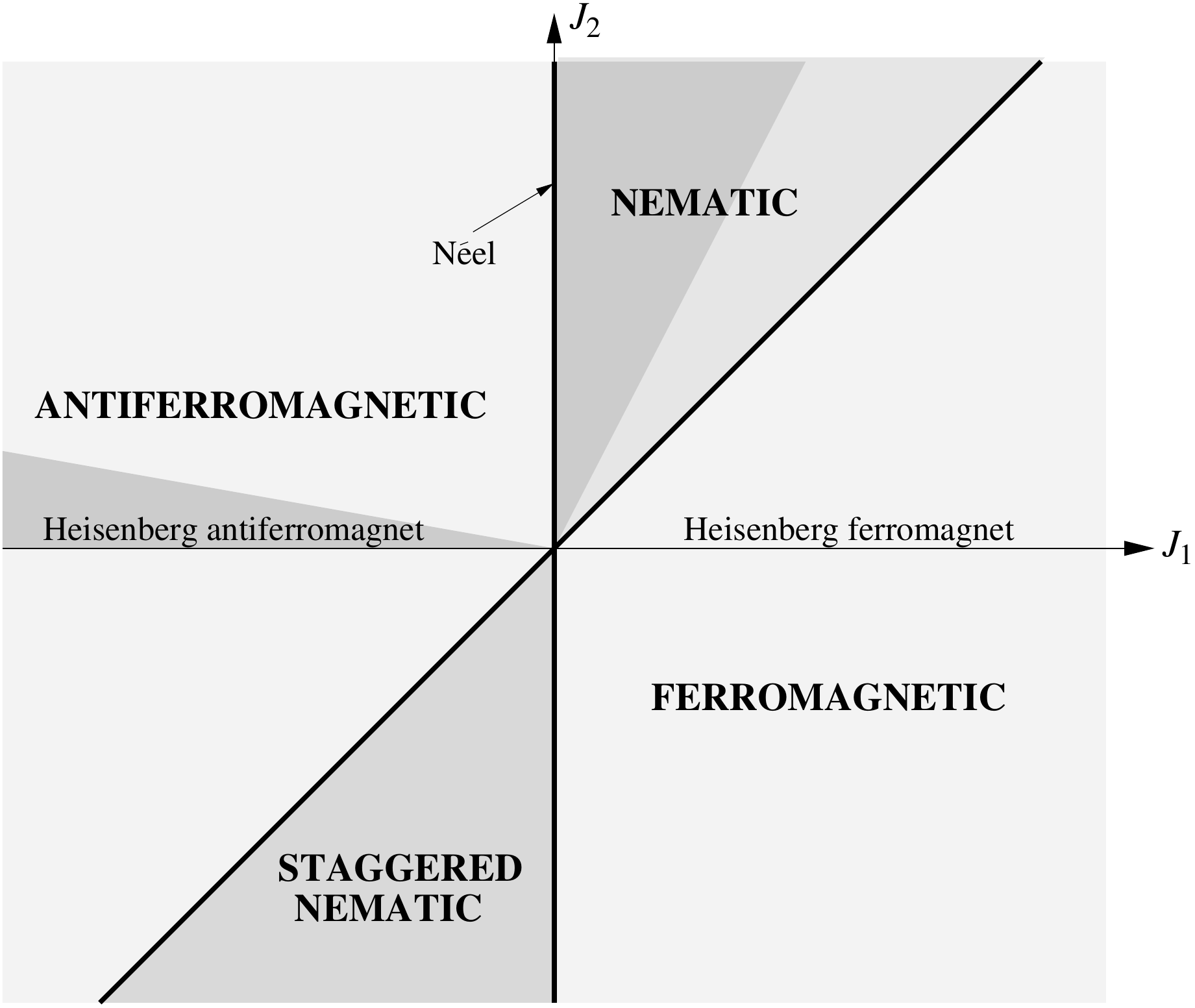}
\caption{Phase diagram of the general spin 1 model with Hamiltonian $H = -\sum [J_{1} \vec S_{x} \cdot \vec S_{y} + J_{2} (\vec S_{x} \cdot \vec S_{y})^{2}]$, in dimension $d\geq3$. The random loop representation applies to the half-quadrant $0 \leq J_{1} \leq J_{2}$. The phase diagram is expected to show four phases (ferromagnetic, spin nematic, antiferromagnetic, staggered spin nematic). This is supported by rigorous results in the dark region around $J_{1}<0$ and small $J_{2}>0$ \cite{DLS,KLS1}, and in the dark region $0 \leq J_{1} \leq \frac12 J_{2}$ \cite{Uel}.}
\label{fig phd}
\efig

\section{Random loop models}

Let us first describe the models of random loops. The connection to quantum spin systems will be described in the next two sections.

At each edge $\{x,y\} \in \caE$ is attached the interval $[0,\beta]$ and a Poisson point measure where ``crosses'' occur with intensity $u$ and ``double bars'' occur with intensity $1-u$. Let $\omega$ denote a realization and $\rho(\dd\omega)$ denote independent Poisson point measures on $\caE \times [0,\beta]$.

To a given realization $\omega$ of the Poisson point measure corresponds a set of loops, denoted $\caL(\omega)$. The loops consist of vertical lines connected by crosses or bars. This is best understood by looking at pictures, see Fig.\ \ref{fig loops}. A mathematically precise definition can be found in \cite{Uel}.

\begin{centering}
\bfig
\begin{picture}(0,0)%
\includegraphics{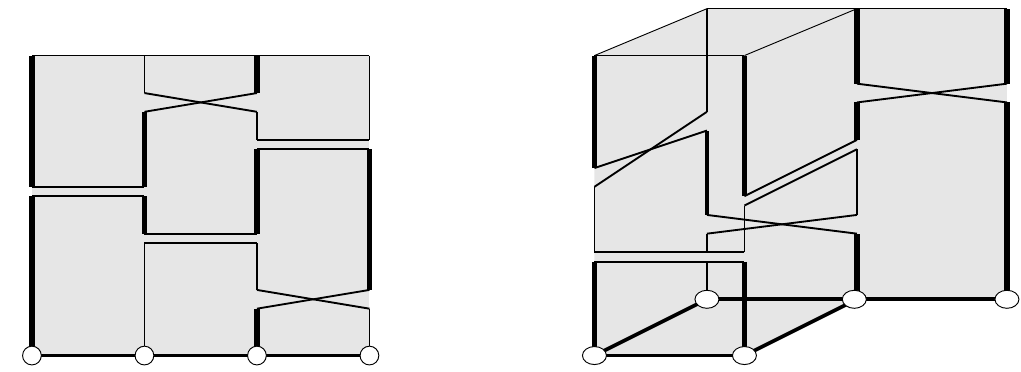}%
\end{picture}%
\setlength{\unitlength}{2368sp}%
\begingroup\makeatletter\ifx\SetFigFont\undefined%
\gdef\SetFigFont#1#2#3#4#5{%
  \reset@font\fontsize{#1}{#2pt}%
  \fontfamily{#3}\fontseries{#4}\fontshape{#5}%
  \selectfont}%
\fi\endgroup%
\begin{picture}(8159,2937)(1246,-4090)
\put(5731,-4036){\makebox(0,0)[lb]{\smash{{\SetFigFont{8}{9.6}{\rmdefault}{\mddefault}{\updefault}{\color[rgb]{0,0,0}$0$}%
}}}}
\put(7876,-3961){\makebox(0,0)[lb]{\smash{{\SetFigFont{8}{9.6}{\rmdefault}{\mddefault}{\updefault}{\color[rgb]{0,0,0}$\Lambda$}%
}}}}
\put(4351,-4036){\makebox(0,0)[lb]{\smash{{\SetFigFont{8}{9.6}{\rmdefault}{\mddefault}{\updefault}{\color[rgb]{0,0,0}$\Lambda$}%
}}}}
\put(1276,-1636){\makebox(0,0)[lb]{\smash{{\SetFigFont{8}{9.6}{\rmdefault}{\mddefault}{\updefault}{\color[rgb]{0,0,0}$\beta$}%
}}}}
\put(5776,-1636){\makebox(0,0)[lb]{\smash{{\SetFigFont{8}{9.6}{\rmdefault}{\mddefault}{\updefault}{\color[rgb]{0,0,0}$\beta$}%
}}}}
\put(1246,-4036){\makebox(0,0)[lb]{\smash{{\SetFigFont{8}{9.6}{\rmdefault}{\mddefault}{\updefault}{\color[rgb]{0,0,0}$0$}%
}}}}
\end{picture}%
\caption{Graphs and realizations of Poisson point measures, and their loops. In both cases, the number of loops is $|\caL(\omega)| = 2$.}
\label{fig loops}
\efig
\end{centering}

The relevant probability distribution involves multiplicative weights with respect to loops. We consider a function $w(\gamma)$ that assigns a real number to each loop $\gamma$. We will consider explicit weights below; for now, we just assume that $w(\gamma)$ depends on the loop in a continuous fashion, so all integrals below are well defined. In the case where $w(\gamma)$ is nonnegative for all $\gamma$ we have a probabilistic setting. But it is useful to include the possibility of negative weights as well.

We define the {\it partition function} as
\be
Y(\beta,\Lambda) = \int \rho(\dd\omega) \prod_{\gamma\in\caL(\omega)} w(\gamma).
\ee
We will always consider cases where $Y(\beta,\Lambda)\neq0$. The relevant measure for the model of random loops is given by
\be
\label{La mais-je}
\frac1{Y(\beta,\Lambda)} \, \Bigl( \prod_{\gamma\in\caL(\omega)} w(\gamma) \Bigr) \, \rho(\dd\omega).
\ee
It is a probability measure when the weights are positive.

It is not hard to show that for $\beta$ small, and under some conditions on $w(\gamma)$, the loops have small lengths and the probability that two sites belong to the same loop shows exponential decay with respect to the distance between the sites. See e.g.\ Theorem 6.1 in \cite{GUW}.

The special case of constant weights, $w(\gamma)=\theta$, is interesting, and actually relevant to quantum systems without external magnetic fields. Under some additional assumptions, namely that 
the graph $(\Lambda,\caE)$ be a $d$-dimensional cube with even side lengths $L$ and $d\geq3$, that $u \in [0,\frac12]$, and that $\theta=2,3,\dots$, one can prove the existence of {\it macroscopic loops} when $\beta$ is sufficiently large. Let $\ell_{0}$ denote the random variable for the length of the loop that contains the point $(0,0) \in \Lambda \times [0,\beta]$. The length of the loop is defined as the sum of the length of all its vertical elements.

\begin{theorem}
\label{thm lolo}
Under the assumptions listed in the paragraph above, there exists $c>0$ such that for all $L$,
\[
\bbE \Bigl( \frac{\ell_{0}}{\beta L^{d}} \Bigr) \geq c.
\]
\end{theorem}

See \cite[Chapter 5]{Uel} for the statement with precise conditions. This theorem can be proved using the method of infrared bounds and reflection positivity introduced and developed in \cite{FSS, DLS, KLS1, AFFS}; see Biskup \cite{Bis} for an excellent survey. This theorem implies the occurrence of long-range order in some quantum systems. The main novel result is the occurrence of spin nematic order for the spin 1 model with Hamiltonian $\tilde H_{\Lambda,\bsh}^{(u)}$ defined in Eq.\ \eqref{def Ham spin 1}. Indeed, it follows from Theorem \ref{thm lolo} that
\begin{itemize}
\item $\displaystyle \frac1{|\Lambda|} \sum_{x} \langle S_{0}^{3} S_{x}^{3} \rangle > c$ for the model $H = -\sum (\vec S_{x} \cdot \vec S_{y})^{2}$, with $c>0$ independent of $\Lambda$, $d\geq5$, $\beta$ large enough. (This is actually N\'eel order.)
\item $\displaystyle \frac1{|\Lambda|} \sum_{x} \bigl( \langle (S_{0}^{3})^{2} (S_{x}^{3})^{2} \rangle - \langle (S_{0}^{3})^{2} \rangle \langle (S_{x}^{3})^{2} \rangle \bigr) > c$ for the model $H = -\sum \bigl( J_{1} \vec S_{x} \cdot \vec S_{y} + J_{2} (\vec S_{x} \cdot \vec S_{y})^{2} \bigr)$, with $c>0$ independent of $\Lambda$, $0 \leq J_{1} \leq \frac12 J_{2}$, $d\geq5$, $\beta$ large enough. (When $J_{1} \lesssim \frac12 J_{2}$ the result holds for $d\geq3$.)
\end{itemize}
It does not seem possible to prove this using the method of infrared bounds and reflection positivity directly for quantum systems. The method in \cite{Uel} consists in studying the model $H_{\Lambda}^{(u)}$, which is not related to $\tilde H_{\Lambda}^{(u)}$ in any obvious way when $u \neq 0,1$. The Gibbs operator $\e{-\beta H_{\Lambda}^{(u)}}$ can be expanded in random loops and ``space-time spin configurations'' (see next section), which gives a sort of classical model that is reflection positive. This allows to prove ``Gaussian domination'', leading to infrared bounds for the Duhamel two-point function. Combining with the Falk-Bruch inequality, as in \cite{DLS, KLS1}, one obtains Theorem \ref{thm lolo}. The results for $\tilde H_{\Lambda}^{(u)}$ are then consequences of the loop representation.

\section{Gibbs operator and partition functions}

The first result is a formula for the Gibbs operator $\e{-\beta H}$ in terms of the Poisson point measure $\rho(\dd\omega)$. To a realization $\omega$ corresponds a sequence $(A_{1},t_{1}), \dots, (A_{n},t_{n})$ where $0 < t_{1} < \dots < t_{n} < \beta$ are the times for the occurrence of events in $\omega$, and $A_{j}$ is the operator $T_{xy}$ if the event of time $t_{j}$ is a cross at $\{x,y\} \in \caE$; $A_{j}$ is the operator $Q_{xy}$ if the event of time $t_{j}$ is a double bar at $\{x,y\}$.

\begin{theorem}
\label{thm Gibbs}
We have
\[
\e{-\beta H_{\Lambda,\bsh}^{(u)}} = \int\rho(\dd\omega) \e{-(\beta-t_{n}) \sum h_{x} S_{x}^{3}} A_{n} \e{-(t_{n}-t_{n-1}) \sum h_{x} S_{x}^{3}} A_{n-1} \dots A_{1} \e{-t_{1} \sum h_{x} S_{x}^{3}}.
\]
\end{theorem}

The same representation applies to the operator $\e{-\beta \tilde H_{\Lambda,\bsh}^{(u)}}$, but with $P_{xy}$ instead of $Q_{xy}$ when double bars occur.
The proof can be done by discretizing the time interval $[0,\beta]$, linearizing the Poisson point measure, grouping terms wisely and invoking the Trotter product formula.

Next, we consider partition functions. Given a loop $\gamma$, we denote by $\ell_{x}(\gamma)$ the length of the vertical element(s) of the loop at site $x\in\Lambda$. We have $0 \leq \ell_{x}(\gamma) \leq \beta$ and, for almost all realizations $\omega$,
\be
\sum_{\gamma \in \caL(\omega)} \sum_{x\in\Lambda} \ell_{x}(\gamma) = \beta |\Lambda|.
\ee

\begin{theorem}
\label{thm Z}
Given $S \in \frac12 \bbN$, let
\[
w(\gamma) = \sum_{a=-S}^{S} \exp\Bigl\{ a \sum_{x\in\Lambda} h_{x} \ell_{x}(\gamma) \Bigr\}.
\]
Then for all $u\in[0,1]$, we have
\[
Z^{(u)}(\beta,\Lambda,\bsh) = \int\rho(\dd\omega) \prod_{\gamma\in\caL(\omega)} w(\gamma).
\]
\end{theorem}

In the case where $h_{x} \equiv 0$, we have $w(\gamma) = 2S+1$ for all loops, and the partition function is equal to $\int (2S+1)^{|\caL(\omega)|} \rho(\dd\omega)$.

The corresponding formula for the model with Hamiltonian $\tilde H_{\Lambda,\bsh}^{(u)}$ is more complicated, as it involves vertical directions of loops. Namely, let us choose an orientation for the loops, and let $\ell_{x}^{+}(\gamma)$ (resp.\ $\ell_{x}^{-}(\gamma)$) denote the vertical length of the elements of $\gamma$ at $x$ that move up (resp.\ that move down). We have $\ell_{x}^{+}(\gamma) + \ell_{x}^{-}(\gamma) = \ell_{x}(\gamma)$. We only state the theorem in the case of integer $S$, as there are inelegant signs when $S$ is half integer.

\begin{theorem}
\label{thm Z tilde}
Given $S \in \bbN$, let
\[
w(\gamma) = \sum_{a=-S}^{S} \exp\Bigl\{ a \sum_{x\in\Lambda} h_{x} \bigl[ \ell_{x}^{+}(\gamma) - \ell_{x}^{-}(\gamma) \bigr] \Bigr\}.
\]
Then for all $u\in[0,1]$, we have
\[
\tilde Z^{(u)}(\beta,\Lambda,\bsh) = \int\rho(\dd\omega) \prod_{\gamma\in\caL(\omega)} w(\gamma).
\]
\end{theorem}

In order to prove Theorems \ref{thm Z} and \ref{thm Z tilde}, we need the concept of {\it space-time configurations}. This is also useful in the calculation of correlation functions. A space-time spin configuration is a function
\be
\sigma : \Lambda \times [0,\beta] \longrightarrow \{-S, -S+1, \dots, S\}.
\ee
such that $\sigma_{x,t}$ is piecewise constant in $t$, for any $x$. Given a realization $\omega$ of the Poisson point measure, let $\Sigma_{\rm per}(\omega)$ denote the set of space-time spin configurations that take constant values along each loop. See Fig.\ \ref{fig spinconfig} for an illustration. Let $\Sigma(\omega)$ denote the set of configurations that are compatible with $\omega$, but without requiring that $\sigma_{x,0} = \sigma_{x,\beta}$. Notice that
\be
|\Sigma_{\rm per}(\omega)| = (2S+1)^{|\caL(\omega)|}.
\ee

\begin{centering}
\bfig
\includegraphics[width=85mm]{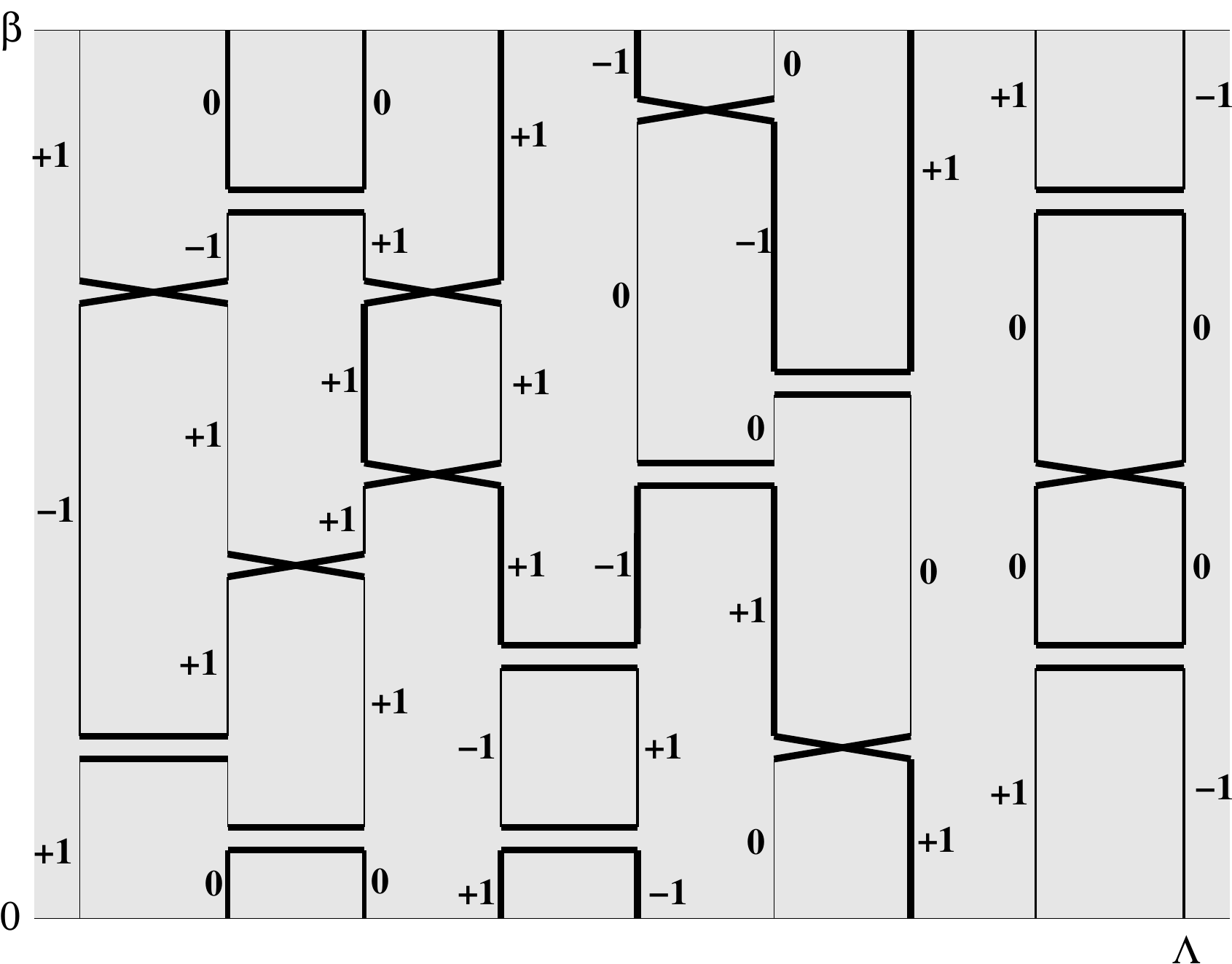}
\caption{Illustration for a realization of the measure $\rho_{\iota}$ and a compatible space-time spin configuration.}
\label{fig spinconfig}
\efig
\end{centering}

We use Theorem \ref{thm Gibbs} and we insert the resolution of the identity $\Id = \sum_{\sigma} |\sigma\rangle \langle\sigma|$ on the left of each transition $A_{j}$. Because of the definitions of $T_{xy}$ and $Q_{xy}$, we get
\be
\begin{split}
\Tr \e{-\beta H_{\Lambda,\bsh}^{(u)}} &= \int\rho(\dd\omega) \sum_{\sigma \in \Sigma_{\rm per}(\omega)} \exp\Bigl\{ -\sum_{x\in\Lambda} \int_{0}^{\beta} h_{x} \sigma_{x,t} \dd t \Bigr\} \\
&= \int\rho(\dd\omega) \prod_{\gamma \in \caL(\omega)} \sum_{a=-S}^{S} \e{-a \sum_{x} h_{x} \ell_{x}(\gamma)}.
\end{split}
\ee
This gives the claim of Theorem \ref{thm Z}. The proof of Theorem \ref{thm Z tilde} is similar but with different sets of space-time spin configurations. Let $\tilde\Sigma(\omega)$, $\tilde\Sigma_{\rm per}(\omega)$ with the prescription that the sign of the spin changes when the vertical direction of the loop changes. The calculation is then the same, with additional signs due to the double bars. Namely, a double bar at $\{x,y\} \times t$ gives the sign $(-1)^{\sigma_{x,t-} + \sigma_{x,t+}}$. Fortunately each loop involves an even number of minus signs, so the weight is positive.

\section{Correlation functions}

We restrict ourselves to two-point correlation functions. We also make the important simplification $h_{x} \equiv 0$, although expressions can certainly be derived for nonzero external magnetic fields. The loop correlations are given by just three events:
\begin{itemize}
\item $E_{x,y}^{+}$ is the set of all realizations $\omega$ such that $x$ and $y$ belong to the same loop, and with identical vertical direction at these points.
\item $E_{x,y}^{-}$ is the set of all $\omega$ such that $x$ and $y$ belong to the same loop, and with opposite vertical directions at these points.
\item $E_{x,y} = E_{x,y}^{+} \cup E_{x,y}^{-}$ is the set of all $\omega$ such that $x$ and $y$ belong to the same loop.
\end{itemize}
These events are illustrated in Fig.\ \ref{fig correl}.

\bfig
\includegraphics[width=120mm]{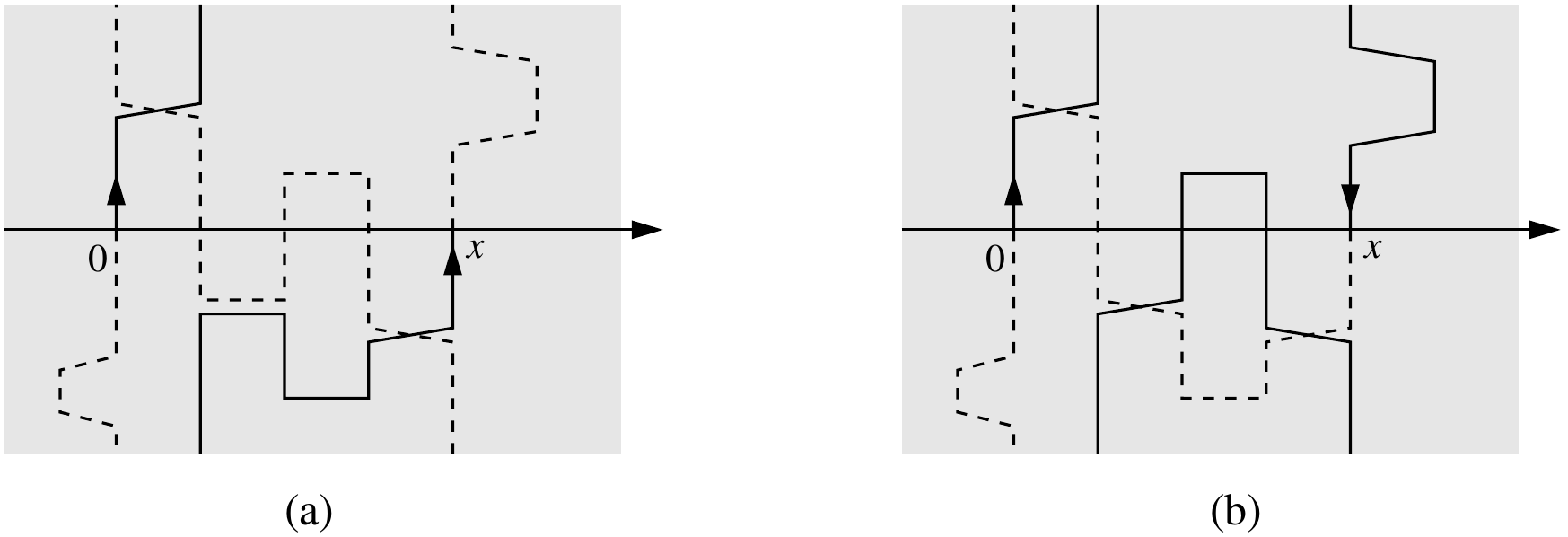}
\caption{Illustration for (a) the event $E_{0,x}^{+}$; (b) the event $E_{0,x}^{-}$.}
\label{fig correl}
\efig

Let $A_{x}$ be an operator of the form $A \otimes \Id_{\Lambda\setminus\{x\}}$ and $B_{y}$ be an operator of the form $B \otimes \Id_{\Lambda\setminus\{y\}}$, where $A,B$ are operators on $\bbC^{2S+1}$. We consider the two-point function
\be
\langle A_{x} B_{y} \rangle = \frac1{Z^{(u)}(\beta,\Lambda,0)} \Tr A_{x} B_{y} \e{-\beta H_{\Lambda,0}^{(u)}}.
\ee
We use the notation $\tr$ for the trace in $\bbC^{2S+1}$. Let $\bbP(\cdot)$ denote the probability with respect to the random loop measure $Z^{(u)}(\beta,\Lambda,0)^{-1} (2S+1)^{|\caL(\omega)|} \rho(\dd\omega)$.

\begin{theorem}
\label{thm corr}
For $x\neq y$, the correlation function above is given by
\[
\langle A_{x} B_{y} \rangle = \tfrac1{2S+1} (\tr AB) \; \bbP(E_{x,y}^{+}) + \tfrac1{2S+1} (\tr AB^{\rm t}) \; \bbP(E_{x,y}^{-}) + \tfrac1{(2S+1)^{2}} (\tr A) (\tr B) \; \bbP(E_{x,y}^{\rm c}).
\]
\end{theorem}

Here, $B^{\rm t}$ denotes the transpose of the matrix $B$ in the basis where $\{S_{x}^{3}\}$ is diagonal. Choosing $B = \Id$, we get the formula for the one-point function (it is relevant for truncated correlation functions):
\be
\langle A_{x} \rangle = \tfrac1{2S+1} \tr A.
\ee
An interesting special case of correlation function is $S=\frac12$ and $A=B=S^{i}$. We find that
\be
\langle S_{x}^{i} S_{y}^{i} \rangle = \begin{cases} \tfrac14 \bbP(E_{x,y}) & \text{if } i=1,3, \\ \tfrac14 [ \bbP(E_{x,y}^{+}) - \bbP(E_{x,y}^{-})] & \text{if } i=2. \end{cases}
\ee

\begin{proof}[Proof of Theorem \ref{thm corr}]
We use Theorem \ref{thm Gibbs} and space-time spin configurations, and we get
\be
\Tr A_{x} B_{y} \e{-\beta H_{\Lambda,0}^{(u)}} = \int\rho(\dd\omega) \sum_{\sigma \in \Sigma(\omega)} \langle \sigma_{\boldsymbol\cdot,0}| A_{x} B_{y} |\sigma_{\boldsymbol\cdot,\beta} \rangle.
\ee
Next we decompose
\be
\int {\boldsymbol\cdot} \, \rho(\dd\omega) = \int_{E_{x,y}^{+}} {\boldsymbol\cdot} \, \rho(\dd\omega) + \int_{E_{x,y}^{-}} {\boldsymbol\cdot} \, \rho(\dd\omega) + \int_{E_{x,y}^{\rm c}} {\boldsymbol\cdot} \, \rho(\dd\omega)
\ee
and we treat each case separately. If $\omega \in E_{x,y}^{+}$, we find
\be
\sum_{\sigma \in \Sigma(\omega)} \langle \sigma_{\boldsymbol\cdot,0}| A_{x} B_{y} |\sigma_{\boldsymbol\cdot,\beta} \rangle = (2S+1)^{|\caL(\omega)|-1} \sum_{a,b=-S}^{S} \langle a,b | A_{x} B_{y} |b,a\rangle.
\ee
The term $(2S+1)^{|\caL(\omega)|-1}$ is due to the sum of spin configurations on all the loops except the one that contains $x$ and $y$. The sum over $a,b$ represents the possible values of spins along this loop. Now we have
\be
\langle a,b | A_{x} B_{y} |b,a\rangle = \langle a | A | b \rangle \langle b | B | a \rangle,
\ee
and the sum over $a,b$ gives $\tr AB$. The case where $\omega \in E_{x,y}^{-}$ is similar, but the matrix elements involving $A,B$ are
\be
\langle a,a | A_{x} B_{y} |b,b\rangle = \langle a | A | b \rangle \langle a | B | b \rangle.
\ee
The sum over $a,b$ gives $\tr AB^{\rm t}$. Finally, the case $\omega \in E_{x,y}^{\rm c}$ involves two special loops, those containing $x$ and $y$, and we get
\be
\sum_{\sigma \in \Sigma(\omega)} \langle \sigma_{\boldsymbol\cdot,0}| A_{x} B_{y} |\sigma_{\boldsymbol\cdot,\beta} \rangle = (2S+1)^{|\caL(\omega)|-2} \tr A \; \tr B.
\ee
\end{proof}

The case of the Hamiltonian $\tilde H^{(u)}_{\Lambda,0}$ is more complicated due to the signs. They lead to signed measures when $S$ is half-integer but not integer (except when $u=0$ on a bipartite lattice). We restrict here to integer $S$ and we consider the two-point function
\be
\langle A_{x} B_{y} \rangle^{\sim} = \frac1{\tilde Z^{(u)}(\beta,\Lambda,0)} \Tr A_{x} B_{y} \e{-\beta \tilde H_{\Lambda,0}^{(u)}}.
\ee
We also write
\be
\tilde \bbP(E_{x,y}) = \frac1{\tilde Z^{(u)}(\beta,\Lambda,0)} \int 1_{E_{x,y}}(\omega) (2S+1)^{|\caL(\omega)|} \, \rho(\dd\omega).
\ee

\begin{theorem}
\label{thm integer spin}
For $x\neq y$, the correlation above is given by
\[
\begin{split}
\langle A_{x} B_{y} \rangle^{\sim} = &\tfrac1{2S+1} (\tr AB) \; \tilde\bbP(E_{x,y}^{+}) + \tfrac1{2S+1} \biggl( \sum_{a,b=-S}^{S} (-1)^{a+b} \langle a | A | b \rangle \langle -a | B | -b \rangle \biggr) \; \tilde\bbP(E_{x,y}^{-}) \\
&+ \tfrac1{(2S+1)^{2}} (\tr A) (\tr B) \; \tilde\bbP(E_{x,y}^{\rm c}).
\end{split}
\]
\end{theorem}

The formula for one-point functions follow, $\langle A_{x} \rangle^{\sim} = \frac1{2S+1} \tr A$.
The proof is similar to that of Theorem \ref{thm corr}, but there are extra difficulties due to the minus signs. We do not write it explicitly. We find in particular (see \cite{Uel} for more details)
\be
\begin{split}
& \langle S_{x}^{i} S_{y}^{i} \rangle = \tfrac13 S(S+1) \, \bigl[ \bbP(E_{x,y}^{+}) - \bbP(E_{x,y}^{-}) \bigr] \\
& \langle (S_{x}^{i})^{2} (S_{y}^{i})^{2} \rangle - \langle (S_{x}^{i})^{2} \rangle \langle (S_{y}^{i})^{2} \rangle = \tfrac1{45} S(S+1)(2S-1)(2S+3) \, \bbP(E_{x,y}).
\end{split}
\ee

It is remarkable that many spin correlation functions can be expressed with a handful of loop correlation functions.

\bigskip
\noindent
{\small
{\bf Acknowledgments:}
I am grateful to Benjamin Lees for useful discussions. I would like to thank Pavel Exner, Wolfgang K\"onig, and Hagen Neidhardt, for organizing the conference QMATH 12 in Berlin, 10--13 September 2013, where this work was presented.
}

%\newpage
{
\renewcommand{\refname}{\small References}
\bibliographystyle{symposium}

}

\end{document}